\documentclass[11pt]{article}

\usepackage[letterpaper, top=1in, bottom=1in, left=1in, right=1.0in, includefoot, marginparwidth=60pt]{geometry}

\usepackage{graphicx,cite,xspace}

\usepackage{amsmath}

\usepackage{amsthm}

\usepackage{amssymb}

\usepackage{amsfonts}
\newtheorem{proposition}{Proposition}
\newtheorem{theorem}{Theorem}

\newtheorem{lemma}{Lemma}
\newtheorem{observation}{Observation}
\newtheorem{claim}{Claim}
\newtheorem{conjecture}{Conjecture}

\newcounter{func}

\newcommand{\funref}[1]{\hyperref[#1]{f_{\ref*{#1}}}}

\newcounter{const}

\newcommand{\cstref}[1]{\hyperref[#1]{c_{\ref*{#1}}}}

\usepackage{color}
\definecolor{black}{rgb}{0, 0, 0}

\definecolor{white}{rgb}{1, 1, 1}

\definecolor{grey}{rgb}{.6, .6, .6}

\definecolor{red}{rgb}{1, 0 ,0}

\definecolor{green}{rgb}{0, 1, 0}

\definecolor{blue}{rgb}{0, 0 ,1}

\definecolor{darkred}{rgb}{0.7, 0 ,0}

\definecolor{darkgreen}{rgb}{0, 0.7, 0}

\definecolor{darkblue}{rgb}{0, 0 , 0.55}

\definecolor{magenta}{rgb}{1, 0, 1}

\definecolor{cyan}{rgb}{0, 1, 1}

\definecolor{yellow}{rgb}{1, 0.9, 0}

\definecolor{purple}{rgb}{0.5, 0, 0.5}

\definecolor{orange}{rgb}{1, 0.5, 0}

\definecolor{turquoise}{rgb}{0, 0.7, 0.7}

\newcommand{\cc}{{\sf cc}}
\newcommand{\bN}{\mathbb{N}}

\newcommand{\FPT}{{\sf FPT}\xspace}
\newcommand{\NP}{{\sf NP}}
\newcommand{\no}{{\sf no}}
\newcommand{\obs}{{\sf obs}}

\newcommand{\yes}{{\sf yes}}
\newcommand{\ii}{{/\!\!/}}

\usepackage[mathscr]{euscript}

\newcommand{\vc}{\mathsf{vc}}
\newcommand{\p}{\mathsf{p}}
\newcommand{\idf}{\mathsf{idf}}
\newcommand{\ec}{\mathsf{ec}}

\newcommand{\ecf}{\mathsf{ecf}}
\newcommand{\fvs}{\mathsf{fvs}}
\newcommand{\id}{\mathsf{id}}
\newcommand{\gall}{\mathcal{G}_{{\text{\rm  \textsf{all}}}}}

\newcommand{\Ccal}{\mathcal{C}}

\newcommand{\Ecal}{\mathcal{E}}
\newcommand{\Fcal}{\mathcal{F}}

\newcommand{\Hcal}{\mathcal{H}}

\newcommand{\Mcal}{\mathcal{M}}

\newcommand{\Ocal}{\mathcal{O}}
\newcommand{\Pcal}{\mathcal{P}}

\newcommand{\Scal}{\mathcal{S}}
\newcommand{\Tcal}{\mathcal{T}}

\newcommand{\Vcal}{\mathcal{V}}

\newcommand{\Xcal}{\mathcal{X}}
\newcommand{\Ycal}{\mathcal{Y}}
\newcommand{\Zcal}{\mathcal{Z}}

\newcommand{\Nbbb}{\mathbb{N}}

\RequirePackage{stmaryrd}
\usepackage{textcomp}
\DeclareUnicodeCharacter{2286}{\subseteq}
\DeclareUnicodeCharacter{2192}{\ifmmode\to\else\textrightarrow\fi}
\DeclareUnicodeCharacter{2203}{\ensuremath\exists}
\DeclareUnicodeCharacter{183}{\cdot}
\DeclareUnicodeCharacter{2200}{\forall}
\DeclareUnicodeCharacter{2264}{\leq}
\DeclareUnicodeCharacter{2265}{\geq}
\DeclareUnicodeCharacter{8614}{\mathbin{\mapsto}}
\DeclareUnicodeCharacter{8656}{\Leftarrow}
\DeclareUnicodeCharacter{8657}{\Uparrow}
\DeclareUnicodeCharacter{8658}{\Rightarrow}
\DeclareUnicodeCharacter{8659}{\Downarrow}
\DeclareUnicodeCharacter{8669}{\rightsquigarrow}
\newcommand{\eqdef}{\stackrel{{\scriptsize\rm def}}{=}}
\DeclareUnicodeCharacter{8797}{\eqdef}
\DeclareUnicodeCharacter{8870}{\vdash}
\DeclareUnicodeCharacter{8873}{\Vdash}
\DeclareUnicodeCharacter{22A7}{\models}
\DeclareUnicodeCharacter{9121}{\lceil}
\DeclareUnicodeCharacter{9123}{\lfloor}
\DeclareUnicodeCharacter{9124}{\rceil}
\DeclareUnicodeCharacter{2208}{\in}
\DeclareUnicodeCharacter{9126}{\rfloor}
\DeclareUnicodeCharacter{9655}{\triangleright}
\DeclareUnicodeCharacter{9665}{\triangleleft}
\DeclareUnicodeCharacter{9671}{\diamond}
\DeclareUnicodeCharacter{9675}{\circ}
\DeclareUnicodeCharacter{10178}{\bot}
\DeclareUnicodeCharacter{10214}{} 
\DeclareUnicodeCharacter{10215}{} 
\DeclareUnicodeCharacter{10229}{\longleftarrow}
\DeclareUnicodeCharacter{10230}{\longrightarrow}
\DeclareUnicodeCharacter{10231}{\longleftrightarrow}
\DeclareUnicodeCharacter{10232}{\Longleftarrow}
\DeclareUnicodeCharacter{10233}{\Longrightarrow}
\DeclareUnicodeCharacter{10234}{\Longleftrightarrow}
\DeclareUnicodeCharacter{10236}{\longmapsto}
\DeclareUnicodeCharacter{10238}{\Longmapsto} 
\DeclareUnicodeCharacter{10503}{\Mapsto}    
\DeclareUnicodeCharacter{10971}{\mathrel{\not\hspace{-0.2em}\cap}}
\DeclareUnicodeCharacter{65294}{\ldotp}
\DeclareUnicodeCharacter{65372}{\mid}

\usepackage{tikz}
\usetikzlibrary{calc}
\usepackage{boxedminipage}
\usepackage{framed}
\usepackage{thm-restate}
\usepackage{tabularx}
\usepackage{tcolorbox}
\usepackage{etoolbox}
\usepackage{xifthen}
\usepackage{ifthen}
\newlength{\RoundedBoxWidth}
\newsavebox{\GrayRoundedBox}
\newenvironment{GrayBox}[1]%
   {\setlength{\RoundedBoxWidth}{.93\textwidth}
    \def\boxheading{#1}
    \begin{lrbox}{\GrayRoundedBox}
       \begin{minipage}{\RoundedBoxWidth}}%
   {   \end{minipage}
    \end{lrbox}
    \begin{center}
    \begin{tikzpicture}%
       \node(Text)[draw=black!20,fill=white,rounded corners,%
             inner sep=2ex,text width=\RoundedBoxWidth]%
             {\usebox{\GrayRoundedBox}};
        \coordinate(x) at (current bounding box.north west);
        \node [draw=white,rectangle,inner sep=3pt,anchor=north west,fill=white]
        at ($(x)+(6pt,.75em)$) {\boxheading};
    \end{tikzpicture}
    \end{center}}
\newenvironment{defproblemx}[2][]{\noindent\ignorespaces%
                                \FrameSep=6pt%
                                \parindent=0pt%
                \vspace*{-1.5em}
                \ifthenelse{\isempty{#1}}{%
                  \begin{GrayBox}{\textsc{#2}}%
                }{%
                  \begin{GrayBox}{\textsc{#2}  parameterized by~{#1}}%
                }
                \begin{tabular*}{\textwidth}{@{\hspace{.1em}} >{\itshape} p{1.8cm} p{0.8\textwidth} @{}}%
            }{
                \end{tabular*}%
                \end{GrayBox}%
                \ignorespacesafterend
            }
\newcommand{\defproblem}[3]{
  \begin{defproblemx}{#1}
    Input:  & #2 \\
    Question: & #3
  \end{defproblemx}
}%

\usepackage[
plainpages = true,
pdfpagelabels,
hyperfootnotes=true,
pdfstartview =,
bookmarks=true,
bookmarksopen = true,
bookmarksnumbered = true,
breaklinks = true,
hyperfigures,
pagebackref,
urlcolor = black,
anchorcolor = green,
hyperindex = true,
colorlinks = true,
linkcolor = blue,
citecolor = turquoise
]{hyperref}

\usepackage[english]{babel}
\addto\extrasenglish{} 
\addto\extrasenglish{} 

\newenvironment{cproof}{\proof[Proof of claim]}{\endproof}

\usepackage{verbatim,listings}
\usepackage{todonotes}
\newcommand*\samethanks[1][\value{footnote}]{\footnotemark[#1]}

\title{Vertex identification to a forest\thanks{The first and the third author were supported by  the French-German Collaboration ANR/DFG Project UTMA (ANR-20-CE92-0027) and the Franco-Norwegian project PHC AURORA 2024 (Projet n° 51260WL). The second author was supported by  ANR project ELIT (ANR-20-CE48-0008).}}

\author{Laure Morelle\thanks{LIRMM, Université de Montpellier, CNRS, Montpellier, France. Emails: \texttt{laure.morelle@lirmm.fr}, \texttt{ignasi.sau@lirmm.fr}, 
\texttt{sedthilk@thilikos.info}\ .} \and Ignasi Sau\samethanks \and Dimitrios M. Thilikos\samethanks}

\date{\today}

\begin{document}
\maketitle

\begin{abstract}
\noindent Let $\Hcal$ be a graph class and $k\in\bN$.
We say a graph $G$ admits a \emph{$k$-identification to $\Hcal$} if there is a partition $\Pcal$ of some set $X\subseteq V(G)$ of size at most $k$ such that after identifying each part in $\Pcal$ to a single vertex, the resulting graph belongs to $\Hcal$.
The graph parameter $\id_{\cal H}$ is defined so that
$\id_{\cal H}(G)$ is the minimum  $k$ such that $G$ admits a $k$-identification to $\Hcal$, and the problem of {\sc Identification to $\Hcal$} asks, given a graph $G$ and $k\in\bN$, whether $\id_{\cal H}(G)≤k$.
If we set $\Hcal$ to be the class $\Fcal$ of acyclic graphs,  we
generate the problem  {\sc Identification to Forest}, which we show to be {\sf NP}-complete.
We prove that, when parameterized by the size $k$ of the
identification set, it admits
a kernel of size $2k+1$. For our kernel we reveal a close  relation of   {\sc Identification to Forest}  with the {\sc Vertex Cover}
problem.
We also study the combinatorics of the \yes-instances
of  {\sc Identification to $\Hcal$}, i.e.,  the class $\Hcal^{(k)}:=\{G\mid \id_{\Hcal}(G)≤k\}$, {which we show to be minor-closed for every $k$} when $\Hcal$ is minor-closed.
We prove that the minor-obstructions of $\Fcal^{(k)}$
are of size at most $2k+4$. We also prove that
every graph $G$ such that $\id_{\Fcal}(G)$ is sufficiently big contains as a minor either a cycle on $k$ vertices, or $k$ disjoint triangles, or the \emph{$k$-marguerite} graph, that is the graph obtained by $k$ disjoint triangles by identifying one vertex of each of them into the same vertex.
\end{abstract}

\bigskip
\noindent{\bf Keywords.} Vertex identification, Forests, Vertex Cover, Graph minors, Parameterized Algorithms, Kernelization, Obstructions, Universal Obstructions.

\section{Introduction}
\label{sec-intro}

A considerable part of parameterized algorithms has been dedicated to the study of \emph{graph modification problems}. The general scheme for a graph modification problem consists of  some modification operation, accompanied by some a measure on the ``cost'' of this modification, and a target property. The question is, given a graph $G$ and a non-negative integer $k$, whether it is possible to apply to $G$ a modification operation with cost at most $k$ so that the resulting graph has the target property.
A graph modification problem can be seen as a way to define some notion of
 ``distance from triviality'' \cite{GuoHN04astru},  where  the distance is expressed by the measure of the modification operation and the triviality is expressed by the target class.  Most graph modification problems are known to be \NP-complete \cite{LewisY80theno,Yannakakis81edged}.  A well-studied graph modification operation is vertex deletion and the most typical measure is the number of vertices to be deleted.
 A general family of  problems of this type is {\sc $\Hcal$-Deletion}
where $\Hcal$ is a graph class and where
we look for a vertex set $S\subseteq V(G)$ of at most $k$ vertices such that $G-S$ is a graph in $\Hcal$. 
There are many
problems that can be expressed in this way and also  
many results
identifying instantiations of $\Hcal$ where {\sc $\Hcal$-Deletion}, parameterized by $k$, admits a fixed-parameter algorithm (in short {\sf FPT}-algorithm), that is, an algorithm running in time $f(k)\cdot |G|^{\Ocal(1)}$, for some function $f$.
Two of the most classical and widely studied problems of this type are  {\sc Vertex cover}, where  $\Hcal$ is be the class of edgeless graphs,  and  {\sc Feedback Vertex Set}, where $\Hcal$ is the class of acyclic graphs.
Alternative measures of the vertex removal modification have been considered in \cite{EibenGHK21} (see also \cite{thilikos2024excluding}).

Another line of research on modification problems is to consider other modification operations.  Such operations may include edge removals or additions (see \cite{CrespelleDFG24asurv} for an extended survey), edge contractions \cite{HeggernesHLP13obtai,HeggernesHLLP14contra},
or other modification  operations such as subgraph complementations \cite{FominGT19modif}. We should stress that the existing
results on the parameterized complexity of such alternative modification operations
are also many but not as abundant as in in the case of vertex deletion.

Another modification operation called \emph{vertex fusion} was introduced in \cite{Comas2009Oct}.
The vertex fusion of a vertex set $S$ in a graph $G$ consists in deleting $S$ from $G$ and adding instead a new vertex $s$ adjacent to every vertex of $G-S$ that was adjacent to a vertex of $S$.
In other words, the set $S$ is fused (or identified, using our terminology) to a single vertex $s$.
This setting has real life applications. 
Indeed, consider a communication network represented by a graph $G$. 
The goal is that the vertices communicate as fast as possible through the edges of $G$.
A natural problem is hence to ask whether it is possible to perform a small amount of modification to $G$ so that its diameter becomes small.
Usually, the modification considered to reduce the diameter of a graph is to add edges.
Instead, this article \cite{Comas2009Oct} proposes the vertex fusion operation.
Such a fusion corresponds to adding a new, more modern and perhaps more expansive, communication network on a small vertex set $S$ that would allow for instantaneous (or just much faster) communication among the nodes in $S$.
The authors prove that, given a graph $G$ and $k,d\in\bN$, asking for a set of size $S$ at most $k$ whose fusion gives a graph of diameter (or radius or eccentricity) at most $d$ is \NP-complete, and $W[1]$-hard parameterized by $k$.

\paragraph{Identification to a graph property.}
In this paper, we reinitiate a study of vertex fusion in a more general setting, and call our modification operation \emph{vertex identification}.

We use $\gall$ for the class of all graphs. Let ${\cal H}\subseteq \gall$ be some graph class, expressing some graph property.
We say that a graph $G$ admits a \emph{$k$-identification to $\Hcal$}
if there is a partition $\Pcal=\{S_{1},\ldots,S_{r}\}$ of some subset $S$ of $k$ vertices of $G$ such that if, for every $i\in[r]$, we identify the vertices of
$S_{i}$ to a single vertex, we obtain a graph in $\Hcal$ (this operation is defined more formally in \autoref{sec-prelim-identification}).
That way, $k$-identification to $\Hcal$ defines a measure of ``distance from triviality'' from the property $\Hcal$. 
Note that the fusion operation of \cite{Comas2009Oct} corresponds to the particular case of our operation where the partition consists of a single set.
This gives rise to the graph parameter $\id_{\Hcal}:\gall\to\Nbbb$ where, given a graph $G$,
$\id_{\Hcal}(G)$ is the minimum $k$ for which $G$ admits a $k$-identification to $\Hcal$. The general problem is the following.

\defproblem{\sc Identification to $\Hcal$}{A graph $G$ and $k\in\bN$.}{Does $G$ admit a $k$-identification to $\Hcal$?}

We say that $G$ admits an \emph{identification to $\Hcal$} if it admits a $k$-identification to $\Hcal$ for some $k\in\bN$.
Although we will not dwell on the subject, those familiar with \emph{quotient graphs} or \emph{homomorphisms} may observe the following equivalence:
$G$ admits an identification to $\Hcal$ if and only if $G$ admits a quotient graph that belongs to $\Hcal$ if and only if there is a surjective homomorphism from $G$ to a graph in $\Hcal$ (sometimes called $\Hcal$-coloring \cite{Bodirsky2012Aug}).
However, we are not aware of any optimization version of graph homomorphism (or graph quotient) to a fixed graph class that would fit our setting.

 Suppose now that $\Hcal$
is some minor-closed property, i.e., $\Hcal$ contains all minors\footnote{A graph $H$ is a \emph{minor} of a graph $G$  if $H$ can be obtained from a subgraph of $G$ after contracting edges.} of its graphs.
We also denote by $\obs(\Hcal)$ the set of minor-minimal graphs that do not belong to $\Hcal$ and observe that $G\in{\cal G}$ iff $G$ does not contain as a minor any of the graphs in $\obs(\Hcal)$. We call the graphs in  $\obs(\Hcal)$ the \emph{minor obstructions} of $\Hcal$. Keep also in mind that, according to  Robertson and Seymour's theorem~\cite{RobertsonS04XX},
if $\Hcal$ is minor-closed, then $\obs(\Hcal)$ is finite.
Also, according to {\cite{RobertsonS12XXII,RobertsonS09XXI,RobertsonS95GMXIII}},
checking whether a graph $H$ is a minor of a graph $G$
can be done in time\footnote{Given two functions $\chi,\psi\colon \mathbb{N}\rightarrow \mathbb{N},$ we write $\chi(n)=\mathcal{O}_{x}(\psi(n))$ to denote that there exists a computable function $f\colon\mathbb{N} \rightarrow \mathbb{N}$ such that $\chi(n)=\mathcal{O}( f(x)\cdot \psi(n)).$}
 $\Ocal_{|H|}(|G|^3)$. This running time has been improved in \cite{KawarabayashiKR11thed} to a quadratic one and very recently in \cite{korhonen2024minor} to an almost linear one, i.e., $\Ocal_{|H|}(|G|^{1+\varepsilon})$.

Our first observation is that the minor-closedness of $\Hcal$ implies that,
for every $k\in \Nbbb$,  the
graph class $\Hcal^{(k)}:=\{G\mid \id_{\Hcal}(G)≤k\}$ is also minor-closed (\autoref{lem_minor}).
Therefore, because of the aforementioned results, for every minor-closed ${\Hcal}$, the problem {\sc Identification to $\Hcal$} admits an \FPT-algorithm, in particular, an algorithm running in time   $\Ocal_{k}(|G|^{1+\varepsilon})$.
Note that this does not contradict the $W[1]$-hardness result of \cite{Comas2009Oct} (even if the model is not the same) because the class of graphs of diameter at most $d$ is not minor-closed.
Unfortunately, given that we have no upper bound on the size
of the set $\obs(\Hcal^{(k)})$, the parametric dependence hidden in the ``$\Ocal_{k}$'' notation is not constructive. Actually, it can become constructive because of the recent results in \cite{SauST2024parame}. However, this dependence  still remains huge and it is an open challenge to design \FPT-algorithms with reasonable parametric dependencies
for particular  instantiations of $\Hcal$. As a first step in this direction, we consider the problem {\sc Identification to Forest},  that is,  {\sc Identification to $\Hcal$} where $\Hcal$ is the class $\Fcal$ of acyclic graphs. {Note that this is the first non-trivial natural minor-closed class that one may consider, as if we take $\Hcal$ to be the class of edgeless graphs, then the problem can be trivially solved in polynomial time.}
As we observe in \autoref{sec_hard_kernel},  {\sc Identification to Forest} is
an \NP-complete problem (see \autoref{lem_NP}).

A  problem that is similar to {\sc Identification to Forest} is  {\sc Contraction to Forest},  asking whether it is possible to \textsl{contract} $k$ edges in a graph $G$ so to obtain an acyclic graph. According to the results by  Heggernes,  van ’t Hof,  Lokshtanov, and  Paul in  \cite{HeggernesHLLP14contra}, this problem
can be solved   in time $4.98^k\cdot |G|^{\Ocal(1)}$.
As edge contractions are special cases of vertex identifications,  if  $(G,k)$ is a \yes-instance of {\sc Contraction to Forest} then $(G,2k)$ is also a \yes-instance of  {\sc Identification to  Forest}. However, vertex identifications may not be edge contractions, and it is certainly possible
that a \yes-instance of {\sc Identification to  Forest} is certified
by the identifications of non-adjacent vertices  that cannot be simulated by  a small number of edge contractions. More generally, if $\Hcal$
is some minor-closed graph class and $\Hcal^{[k]}$ is the set of all
graphs 
containing an edge set of size at most~$k$ whose contraction 
creates a graph in $\Hcal$, 
then $\Hcal^{[k]}$ is {\sl not} necessarily a minor-closed graph class, for any $k≥1$ (see \autoref{sec-discussion-and-open}). This indicates  that
the identification operation behaves better than the contraction operation from the structural   point of view, and this motivates the definition and study of  {\sc Identification to $\Hcal$} for minor-closed $\Hcal$'s.
To the authors' knowledge, no study of {\sc Identification to $\Hcal$} has been done from the parameterized complexity point of view{, for any instantiation of $\Hcal$}.

\paragraph{A linear kernel.} Our first result is to prove that {\sc Identification to Forest} admits a linear kernel of size $2k+1$. In formal terms, we prove the following:

\begin{theorem}
\label{mnaiol_th}
There is a algorithm that, given an instance $(G,k)$ of {\sc Identification to Forest}, outputs in polynomial time an equivalent instance $(G',k')$
where $|G'|≤2k+1$ and $k'≤k+1$.
\end{theorem}

 The algorithm  of \autoref{mnaiol_th} is based on a structural result
revealing a strong connection between  {\sc Identification to Forest} and the {\sc Vertex Cover} problem. We use $\idf$ as a shortcut of the graph parameter $\id_{\Fcal}$ (recall that $\Fcal$ is the class of forests),
and we use $\vc(G)$ for the minimum size of a vertex cover of $G$, i.e.,
the minimum number of vertices whose removal from $G$ yields an edgeless graph. Given a graph $G$, we denote by $G^{\sf b}$ the graph obtained from $G$ after removing all bridges (edges whose removal increases the number of components).
The relation between $\vc$ and $\idf$ is given by the fact
that, for every graph $G$, $\idf(G)=\vc(G^{\sf b})$ (\autoref{cor_id_vc}).
We also prove that, for every graph $G$, there is a bridgeless
graph $G'$ on $|G|+1$ vertices such that $\vc(G)=\vc(G')$  (\autoref{VCbless}).
\autoref{mnaiol_th} follows as a consequence of these two facts and the known kernelization algorithm for {\sc Vertex Cover} (see \autoref{sec_hard_kernel}).

\paragraph{Obstructions for $\Fcal^{(k)}$.}

Recall that, for every $k\in\Nbbb$, $\Fcal^{(k)}$ is defined as the set of all graphs that admit a $k$-identification to a forest. Clearly, $\Fcal^{(k)}$
is determined by the finite set $\obs(\Fcal^{(k)})$. Identifying $\obs(\Fcal^{(k)})$, for every $k$, requires an upper bound on the size of its elements. This upper bound is not given by the general result of \cite{RobertsonS04XX}. Our next result is to provide such a bound.

\begin{theorem}\label{lem_2}
Let $k\in\bN$.
For any obstruction $G\in\obs(\Fcal^{(k)})$, $|V(G)|\le 2k+4$.
\end{theorem}

A linear upper bound as the above is  known for the obstructions
of the class $\Vcal_{k}=\{G\mid \vc(G)≤k\}$: Dinneen and Lai
proved that $2k+2$ is an upper bound on the size of the graphs in $\obs(\Vcal_{k})$ \cite{DinneenL07propert,Dinneen97toomany}.
The proof of \autoref{lem_2}
is based on a procedure to construct all obstructions of $\Fcal^{(k)}$ using the obstructions of $\Vcal_{k}$ as a starting point.
Then \autoref{lem_2} follows by the upper bound in \cite{DinneenL07propert}.

\paragraph{Universal obstruction of $\idf$.}
A \emph{parametric graph} is a minor-monotone sequence $\mathscr{G}=\langle \mathscr{G}_k \rangle_{k\in\mathbb{N}}$ of graphs, i.e., for every $k\in\mathbb{N}$, $\mathscr{G}_{k}\leq \mathscr{G}_{k+1}$, where `$\leq$' denotes the minor relation.
We say that two parametric graphs $\mathscr{G}^1$ and $\mathscr{G}^2$ are \emph{comparable} if every graph in $\mathscr{G}^1$ is a minor of a graph in $\mathscr{G}^2$ or every graph in $\mathscr{G}^1$ is a minor of a graph in $\mathscr{G}^2$.
Given a minor-monotone\footnote{We say that a graph parameter $\p:\gall\to\Nbbb$ is \emph{minor-monotone}
if, for every two graphs $G$, $G'$, if $G$ is a minor of $G'$, then $\p(G)≤\p(G')$.}
graph parameter $\p:\gall\to\Nbbb$,
and a finite set $\frak{G}=\{\mathscr{G}^1,\ldots,\mathscr{G}^r\}$ of pairwise non-comparable parametric graphs, we say that $\frak{G}$ is a \emph{universal obstruction} of $\p$
if there is a function $f:\Nbbb\to\Nbbb$ (we refer to $f$ as the \emph{gap function}) such that
\begin{itemize}
\item
for every $k\in\Nbbb$, if $G$ excludes all graphs in $\{\mathscr{G}^1_{k},\ldots,\mathscr{G}^r_{k}\}$ as a minor, then  $\p(G)≤f(k)$.
\item $\p(\mathscr{G}^j_{k})≥f(k)$,  for every $j\in[r]$.
\end{itemize}

Universal obstructions serve
as asymptotic characterizations of graph parameters,
as they identify the typical patterns of graphs that should appear whenever the value of a parameter becomes sufficiently big.
Several structural dualities on graph parameters can be described
using universal obstructions, and it has been conjectured
that for every minor-monotone parameter  there always exists some \emph{finite}
universal obstruction \cite{PaulPT2023graph}. (For a survey on universal obstructions see \cite{PauPTl2023universal}.)

Let us give two examples of universal obstructions.
A universal obstruction for $\vc$ is the set $\{\langle k\cdot K_{2} \rangle_{k\in\mathbb{N}}\}$\footnote{For a graph $H$, we denote by $k\cdot H$ the union of $k$ disjoint copies of $H$.} with  linear gap function $f(k)=\Ocal(k)$.
Another example is the universal obstruction
for the parameter $\fvs$, where $\fvs(G)$ is the minimum size of a vertex set of $G$ whose removal yields an acyclic graph.
An interpretation of the Erdős-Pósa's theorem~\cite{ErdosP65inde} is that
 $\{\langle k\cdot K_{3} \rangle_{k\in\mathbb{N}}\}$ is a universal obstruction for $\fvs$ with gap function $f(k)=\Ocal(k\cdot \log k)$. Notice that $\idf$ can be seen as the analogue of $\fvs$ where now, instead of removing vertices, we pick a set of vertices and apply identifications to them.

 Our next result is a universal obstruction for
 $\idf$. We use $C_{k}$ for the cycle on $k$ vertices
 and $k*K_3$ for the \emph{$k$-marguerite} graph, that is, the graph obtained from $k\cdot K_{3}$ by selecting one vertex from each connected component and identifying all selected vertices into a single one.

\begin{theorem}\label{th_univ_obs_idf_more}
The set $\{\langle k\cdot K_{3} \rangle_{k\in\mathbb{N}},\langle C_{k}\rangle_{k\in\Nbbb},\langle k*K_3 \rangle_{k\in\mathbb{N}}\}$ is a universal obstruction of $\idf$, with gap function $f(G)=\Ocal(k^4\cdot\log^2 k)$.
\end{theorem}

\paragraph{Organization of the paper.} In \autoref{sec-prelim}
 we provide some preliminaries and basic observations about the identification operation. In \autoref{sec_hard_kernel} we prove the \NP-completeness and provide a linear kernel for {\sc Identification to Forest} parameterized by the solution size. In~\autoref{sec:obs} we bound the size of the obstructions of $\Fcal^{(k)}$. Finally, in~\autoref{sec:univ_obs} we find the universal obstructions of $\Fcal^{(k)}$.

\section{Preliminaries}
\label{sec-prelim}

\paragraph{Sets and integers.} We denote by $\mathbb{N}$ the set of non-negative integers. Given two integers $p, q,$ where $p \leq q,$ we denote by $[p, q]$ the set $\{p, \dots, q\}.$ For an integer $p \geq 1,$ we set $[p] = [1, p]$ and $\mathbb{N}_{\geq p} = \mathbb{N} \setminus [0, p - 1].$ For a set $S,$ we denote by $2^{S}$ the set of all subsets of $S$ and by $S \choose 2$ the set of all subsets of $S$ of size $2.$ 

\subsection{Basic concepts on graphs} A graph $G$ is a pair $(V, E)$ where $V$ is a finite set and $E \subseteq {V \choose 2},$ i.e., all graphs in this paper are undirected, finite, and without loops or multiple edges.
We refer the reader to~\cite{diestel2016graph}  for any undefined terminology on graphs.
For
an edge $\{x,y\}$, we use  the simpler notation $xy$ (or $yx$).
We also define $V(G) = V$ and $E(G) = E.$ Given $A,B\subseteq V(G)$, we also denote by $E_G(A,B)$ the set of edges of $G$ with one endpoint in $A$ and the other in $B$.
Given a vertex $v \in V(G),$ we denote by $N_{G}(v)$ the set of vertices of $G$ that are adjacent to $v$ in $G.$
Also, given a set $S \subseteq V(G),$ we set $N_{G}(S) = \bigcup_{v \in S} N_{G}(v) \setminus S$.
For $S \subseteq V(G),$ we set $G[S] = (S, E \cap {S \choose 2})$ and use $G - S$ to denote $G[V(G) \setminus S].$ We say that $G[S]$ is an \emph{induced \emph{(}by $S$\emph{)} subgraph} of $G$.
We denote by $\cc(G)$ the connected components of $G$.
A \emph{bridge} (resp. \emph{cut~vertex}) in $G$ is an edge (resp. a vertex) whose removal increases the number of connected components of $G$.
Given $k\in\bN_{\ge1}$, we say that a graph $G$ is \emph{$k$-connected} if, for any set $X$ of size at most $k-1$, $G-X$ is connected.
Given two graphs $G_{1}$ and $G_{2},$ we denote $G_{1}\cup G_{2}=(V(G_{1})\cup V(G_{2}),E(G_{1})\cup E(G_{2})).$

\paragraph{Minors.}
The \emph{contraction} of an edge $e = uv \in E(G)$ results in a graph $G/e$ obtained from $G \setminus \{ u, v \}$ by adding a new vertex $w$ adjacent to all vertices in the set $N_{G}(\{u,v\}).$
Vertex $w$ is called the \emph{heir} of $e$.
A graph $H$ is a \emph{minor} of a graph $G$ if $H$ can be obtained from a subgraph of $G$ after a series of edge contractions.
Equivalently,  $H$ is a minor of $G$ if there is a collection $\Scal=\{S_v\mid v\in V(H)\}$ of pairwise-disjoint connected subsets of $V(G)$ such that, for each edge $xy\in E(H)$, the set $S_x\cup S_y$ is connected in $V(G)$. $\Scal$ is called a \emph{model} of $H$ in $G$.

We say that a graph class $\Hcal$ is \emph{hereditary} (resp. \emph{monotone})
if it contains all the induced subgraphs (resp. subgraphs) of its graphs.
A class $\Hcal$ is \emph{closed under disjoint union}
if it contains the disjoint union of every two of its graphs.
Finally, $\Hcal$ is \emph{closed under 1-clique-sums} if it is closed under disjoint union and, for any $G,G'\in\Hcal$, the graph obtained by identifying a vertex of $G$ with a vertex of $G'$ also belongs to $\Hcal$.

\subsection{Identification operation}
\label{sec-prelim-identification}

\paragraph{Partitions.}
Given $p\in\bN$, a \emph{$p$-partition} of a set $X$ is a set $\{X_1,\ldots,X_p\}$ of non-empty pairwise disjoint subsets of $X$ such that $X=\bigcup_{i\in[p]}X_i$.
A \emph{partition} of $X$, denoted by $\Pcal(X)$, is a $p$-partition of $X$ for some $p\in\bN$.
Given two sets $X,A$, and $\Xcal=\{X_1,\ldots,X_p\}\in\Pcal(X)$, $\Xcal\cap A$ denotes the partition $\{X_1\cap A,\ldots,X_p\cap A\}$ of $X\cap A$.
Given two disjoint sets $X$ and $Y$, and $\Xcal=\{X_1,\ldots,X_p\}\in\Pcal(X)$ and $\Ycal=\{Y_1,\ldots,Y_q\}\in\Pcal(Y)$, $\Xcal\cup\Ycal$ denotes the partition $\{X_1,\ldots,X_p,Y_1,\ldots,Y_q\}\in\Pcal(X\cup Y).$
Given a graph $G$, we define $\Pcal(G):=\{\Xcal\in\Pcal(X)\mid X\subseteq V(G)\}$.
Given $\Xcal=\{X_1,\ldots,X_p\}\in\Pcal(G)$, we set $\bigcup\Xcal:=\bigcup_{i\in[p]}X_i$, and the {order}
of $\Xcal$ is the size of $\bigcup\Xcal$.

Let $G$ be a graph and $X\subseteq V(G)$.
The \emph{identification of $X$} in $G$, denoted by $G\ii X$, is the {result of the} operation that transforms $G$ into a graph $G'$ obtained from $G$ by deleting $X$  and adding instead a new vertex $x$ adjacent to every vertex in $N_G(X)$.
The vertex $x$ is called the \emph{heir} of $X$.
Note that, if $X=\{u,v\}$ with $uv\in E(G)$, then this corresponds to the contraction of $uv$.

Let $\Xcal=\{X_1,...,X_p\}\in\Pcal(G)$.
The \emph{identification of $\Pcal$} in $G$ is the graph $G\ii \Pcal:=G\ii X_1\ii X_2\ii ...\ii X_p$.
Note that the ordering 
of the  members of the partition does not matter in this definition.

\paragraph{Identification to $\Hcal$.}
Let $\Hcal$ be a graph class and $G$ be a graph.
We say that a partition $\Xcal\in\Pcal(G)$ is an \emph{id-$\Hcal$ partition of $G$} if $G\ii\Xcal\in\Hcal$.
A {{\emph{minimum}} id-$\Hcal$ partition} of $G$ is an id-$\Hcal$ partition of $G$ of minimum order.
As explained in the introduction,
the problem of {\sc Identification to $\Hcal$}  asks,
given a graph $G$ and a non-negative integer  $k$,
whether  $G$ admits an id-$\Hcal$ partition of order at most $k$. We denote by $\Hcal^{(k)}$ the set of graphs that admit an id-$\Hcal$ partition of order $k$.

\subsection{Minor-closedness}
\label{minor_obs}
As said in the introduction, identifications preserve minor-closedness.

\begin{lemma}\label{lem_minor}
If $\Hcal$ is a minor-closed graph class, then for every  $k\in\Nbbb$, the class $\Hcal^{(k)}$ is minor-closed.
\end{lemma}

\begin{proof}
Let $G\in\Hcal^{(k)}$ and $H$ be a minor of $G$. Let us show that $H\in\Hcal^{(k)}$.
Let $\Xcal=\{X_1,\dots,X_p\}\in\Pcal(G)$ be an id-$\Hcal$ partition of $G$.
Given that $H$ is a minor of $G$, there is a model $\Scal=\{S_v\mid v\in V(H)\}$ of $H$ in $G$.
For $i\in[p]$, let $Y_i=\{v\in V(H)\mid S_v\cap X_i\ne\emptyset\}$.
Let $Y:=\bigcup_{i\in[p]}Y_i$.
Note that $|Y|\le|\bigcup\Xcal|\le k$. We want to show that the
partition $\Zcal$ of $Y$ induced by the $Y_i$s is an id-$\Hcal$ partition of $H$.
However, it is possible that $Y_i\cap Y_j\ne\emptyset$ for distinct $i,j\in[p]$, so $(Y_1,\dots,Y_p)$ is not a partition of $Y$.
The correct partition is $\Zcal=(Z_1,\dots,Z_q)\in\Pcal(Y)$ defined by merging the $Y_i$'s that intersect.
In other words, for each $i\in[p]$, there is $j\in[q]$ such that $Y_i\subseteq Z_j$, and if $Z_j\setminus Z_i\ne\emptyset$, then there exists $i'\in[p]$ such that $Y_{i'}\subseteq Z_j$ and $Y_i\cap Y_{i'}\ne\emptyset$.
Then $\{S_v\mid v\in V(H)\setminus Y\}\cup\bigcup_{j\in[q]}\{\bigcup_{v\in Z_j} S_v\setminus X\cup\bigcup_{i\in[p],Y_i\subseteq Z_j}\{x_i\}\}$ is a model of $H\ii\Zcal$ in $G\ii\Xcal$, where $x_i$ is the heir of $X_i$.
Given that $\Hcal$ is minor-closed, $H\ii\Zcal\in\Hcal$, and therefore $H\in\Hcal^{(k)}$.
\end{proof}

\section{Hardness result and kernel}
\label{sec_hard_kernel}

{In this section we exploit the relation between {\sc Identification to Forest} and {\sc Vertex Cover} to present a hardness result and a linear kernel for {\sc Identification to Forest}, building on the corresponding results for {\sc Vertex Cover}.

\subsection{Dealing with bridges}
We present a series of observations concerning $k$-identifications.

\begin{observation}\label{obs:disconnect}
Let $\Hcal$ be a {hereditary}  graph class 
and $G$ be a graph.
Then, for every $\Xcal\in\Pcal(G)$, if $G\ii\Xcal\in\Hcal$, then for each $H\in\cc(G)$, $H\ii(\Xcal\cap V(H))\in\Hcal$.
\end{observation}

\begin{proof}
Let $H\in\cc(G)$.
Given that $H\ii(\Xcal\cap V(H))=G\ii\Xcal-(V(G)\setminus V(H))$ and that $\Hcal$ is hereditary, we conclude that $H\ii(\Xcal\cap V(H))\in\Hcal$.
\end{proof}

\begin{observation}\label{obs:union}
Let $\Hcal$ be a graph class that is closed under disjoint union
and $G$ be a graph.
Then, for each $H\in\cc(G)$ and for each $\Xcal_H\in\Pcal(H)$, if $H\ii\Xcal_H\in\Hcal$, then $G\ii\bigcup_{H\in\cc(G)}\Xcal_H{\in \Hcal}$.
\end{observation}

\begin{proof}
Given that $\Hcal$ is closed {under disjoint union} and
$G\ii\bigcup_{H\in\cc(G)}=\bigcup_{H\in\cc(G)}H\ii\Xcal_H$, we conclude that {$G\ii\bigcup_{H\in\cc(G)}\Xcal_H {\in \Hcal}$}.
\end{proof}

\begin{lemma}\label{Feqbridge}
Let $G$ be a graph and $G^{\sf b}$ be the graph obtained from $G$ after removing all bridges.
Then $\idf(G)=\idf(G^{\sf b})$.
\end{lemma}

\begin{proof}
Let $k:=\idf(G)$.
By definition, $G\in\Fcal^{(k)}$.
By \autoref{lem_minor}, $\Fcal^{(k)}$ is minor-closed, so $G-e\in\Fcal^{(k)}$ for any edge $e$ of $G$.
Therefore, $\idf(G-e)\leq \idf(G)$.

By \autoref{obs:disconnect} and \autoref{obs:union}, we may assume without loss  of generality that $G$ is connected.
Let $e$ be a bridge of $G$.
Let $G_1$ and $G_2$ be the two connected components of $G-e$.
For $i\in[2]$, let $\Xcal_i\in\Pcal(X_i)$ be a minimum id-$\Fcal$ partition of $G_i$.
By \autoref{obs:union}, $(G-e)\ii\Xcal\in\Fcal$ where $\Xcal=\Xcal_1\cup\Xcal_2$.
Suppose toward a contradiction that $G\ii\Xcal$ contains a cycle $C$.
Then, given that $(G-e)\ii\Xcal$ is acyclic, it implies that $e$ is an edge of $C$.
Given that no part of $\Xcal$ contains vertices of both $G_1$ and $G_2$, it implies that $e$ is already an edge of a cycle. This contradicts the fact that $e$ is a bridge. The lemma follows by repeatedly applying this argument as long as there is a bridge.
\end{proof}

\begin{lemma}\label{VCeq}
Let $G$ be a bridgeless graph.
Then $\idf(G)=\vc(G)$.
\end{lemma}

\begin{proof}
Let $X$ be a vertex cover of $G$.
Then $G\ii X$ is a star (if $G$ is edgeless, a vertex is considered as a star).
Hence, $\{X\}\in\Pcal(G)$ is an id-$\Fcal$ partition of $G$.
So $\idf(G)\leq \vc(G)$.

Let $\Xcal$ be an id-$\Fcal$ partition of $G$.
Let $F:=G\ii\Xcal\in\Fcal$.
Let us color red the vertices of $F$ that are heirs of a part of $\Xcal$, and blue the other vertices of $F$.
Given that $G$ is bridgeless, $F$ contains no edge whose endpoints are both blue.
Hence, the red vertices form a vertex cover of $F$.
Therefore, $X$ is a vertex cover of $G$.
So $\vc(G)\leq \idf(G)$.
\end{proof}

Then we get the main result of this section as a direct corollary of \autoref{Feqbridge} and \autoref{VCeq}.

\begin{lemma}\label{cor_id_vc}
Let $G$ be a graph and $G^{\sf b}$ be the graph obtained from $G$ after removing all bridges.
Then $\idf(G)=\vc(G^{\sf b})$.
\end{lemma}

\subsection{{\sf NP}-completeness}

Before proving the \NP-completeness of  {\sc Identification to Forest}, we first need the following lemma.

\begin{lemma}\label{VCbless}
Let $k\in\bN$.
Let $G$ be a graph.
Then there is a bridgeless graph $G'$ with $|V(G)|+1$ vertices such that $G\in\Vcal_k$ if and only if $G'\in\Vcal_{k+1}$.
Moreover, $G'$ can be constructed in linear time.
\end{lemma}

\begin{proof}
Let us construct a bridgeless graph $G'$ from $G$.
Let $I$ be the set of isolated vertices of $G$.
We add a new vertex $v$ to $G$ and add an edge between $v$ and every vertex of $G-I$. The constructed graph $G'$ is clearly bridgeless.

Let us check that $G\in\Vcal_k$ if and only if $G'\in\Vcal_{k+1}$.
We assume that $G$ has at least one edge, otherwise the claim is trivially true.
Suppose the $G\in\Vcal_k$ and let $X$ be a vertex cover of $G$ of size at most $k$.
Then $X\cup\{v\}$ is a vertex cover of $G$ so $G'\in\Vcal_{k+1}$.
Suppose now that $G'\in\Vcal_{k+1}$.
Let $Y$ be a vertex cover of $G'$ of size at most $k+1$.
If $v\in Y$, then $Y\setminus \{v\}$ is a vertex cover of $G$ of size at most $k$.
Otherwise $v\notin Y$.
It implies that $V(G)\setminus I\subseteq Y$.
But then, for any vertex $x$ of $G-I$, $N_{G-I}(x)\subseteq Y$.
Therefore, $Y\setminus \{x\}$ is a vertex cover of $G$ of size at most $k$.
Hence, $G\in\Vcal_k$.

Given that $G'$ can be constructed in linear time, the result follows.
\end{proof}

\begin{lemma}\label{lem_NP}
{\sc Identification to Forest} is \NP-complete.
\end{lemma}

\begin{proof}
Given a graph $G$ and a partition $\Xcal\in\Pcal(G)$, checking that $G\ii\Xcal\in\Fcal$ can obviously be done in linear time.
We reduce from {\sc Vertex Cover} that is \NP-hard \cite{Karp10redu}.
Let $G$ be a graph.
Let $b$ be the number of bridges in $G$.
If $b\geq 1$, by \autoref{VCbless}, there is a graph $G'$ with $|V(G)|+1$ vertices such that $G\in\Vcal_k$ if and only if $G'\in\Vcal_{k'}$ where $k':=k+1$.
If $b=0$, we set $G':=G$ and $k':=k$.
Since $G'\in\Vcal_{k'}$ is bridgeless, by \autoref{VCeq}, $G'\in\Fcal^{(k')}$.
Since $G'$ can be constructed in linear time and that $G\in\Vcal_k$ if and only if $G'\in\Fcal^{(k')}$, the result follows.
\end{proof}

\subsection{A kernel for {\sc Identification to Forest}}

The following kernelization result is known for {\sc Vertex Cover}.

\begin{proposition}[\!\!\cite{LiZ18aker,NemhauserT75vertex}]
\label{prop_kernel_VC}
Given an instance $(G,k)$ of {\sc Vertex Cover}, one can compute in polynomial time   an equivalent instance $(G',k')$ such that $|V(G')|\le2k'\le2k$.
\end{proposition}

Hence, we can derive the following kernalization result for {\sc Identification to Forest}.

\begin{lemma}\label{lem_kernel}
Given an instance $(G,k)$ of {\sc Identification to Forest}, one can compute  in polynomial time an equivalent instance $(G',k')$ such that $|V(G')|\le2k+1$ and $k'\le k+1$.
\end{lemma}

\begin{proof}
Let $G_1$ be obtained from $G$ after removing bridges.
By \autoref{cor_id_vc}, $(G,k)$ is a \yes-instance of {\sc Identification to Forest} if and only if $(G_1,k)$ is a \yes-instance of {\sc Vertex Cover}.
By \autoref{prop_kernel_VC}, there is $(G_2,k_2)$ with $|V(G_2)|\le2k_2\le2k$ such that $(G_1,k)$ is a \yes-instance of {\sc Vertex Cover} if and only if $(G_2,k_2)$ is a \yes-instance of {\sc Vertex Cover}.
By \autoref{VCbless}, there is $(G_3,k_3)$ such that $G_3$ is a bridgeless graph with $|V(G_3)|\le|V(G_2)|+1$ and $k_3\le k_2+1$ such that $(G_2,k_2)$ is a \yes-instance of {\sc Vertex Cover} if and only if $(G_3,k_3)$ is a \yes-instance of {\sc Vertex Cover}.
Finally, by \autoref{cor_id_vc}, given that $G_3$ is bridgeless, by \autoref{VCeq}, $(G_3,k_3)$ is a \yes-instance of {\sc Vertex Cover} if and only if it is a \yes-instance of {\sc Identification to Forest}.
Hence the result.
\end{proof}

\section{Obstructions}
\label{sec:obs}

Given that {\sc Vertex Cover} and {\sc Identification to Forest} are strongly related, it is reasonable to suspect that this holds for their obstructions as well.
Already, as a direct corollary of \autoref{cor_id_vc}, we have the two following results.

\begin{observation}\label{obs1}
Let $k\in\bN$ and $F\in\obs(\Fcal^{(k)})$.
Then $F$ is bridgeless.
\end{observation}

\begin{lemma}\label{lem_bridgeless_obs}
Let $k\in\bN$.
The bridgeless obstructions of $\Vcal_k$ are obstructions of $\Fcal^{(k)}$.
\end{lemma}

\begin{proof}
Let $H\in\obs(\Vcal_k)$ be bridgeless.
By \autoref{VCeq}, $H\notin\Fcal^{(k)}$.
Thus, there is a minor $H'$ of $H$ such that $H'\in\obs(\Fcal^{(k)})$.
By \autoref{obs1}, $H'$ is bridgeless.
Therefore, by \autoref{VCeq}, $H'\notin\Vcal_k$.
Given that $H'$ is a minor of $H$ and that $H\in\obs(\Vcal_k)$, we conclude that $H=H'$.
Therefore, $H\in\obs(\Fcal^{(k)})$.
\end{proof}

We are actually going to prove in \autoref{sec:VC_obs} that the only bridges that may occur in an obstruction of $\Vcal_k$ are isolated edges.
Then, in \autoref{sec:id_obs}, we will prove that any obstruction of $\Fcal^{(k)}$ can be obtained from an obstruction of $\Vcal_k$ by adding edges.
See \autoref{fig_tab} for a comparison of $\Vcal_k$ and $\Fcal^{(k)}$ for $k\le3$, where the obstructions of $\Vcal_k$ are taken from \cite{CattellD94achar}.

\begin{figure}[h]
\centering
\includegraphics[scale=0.8]{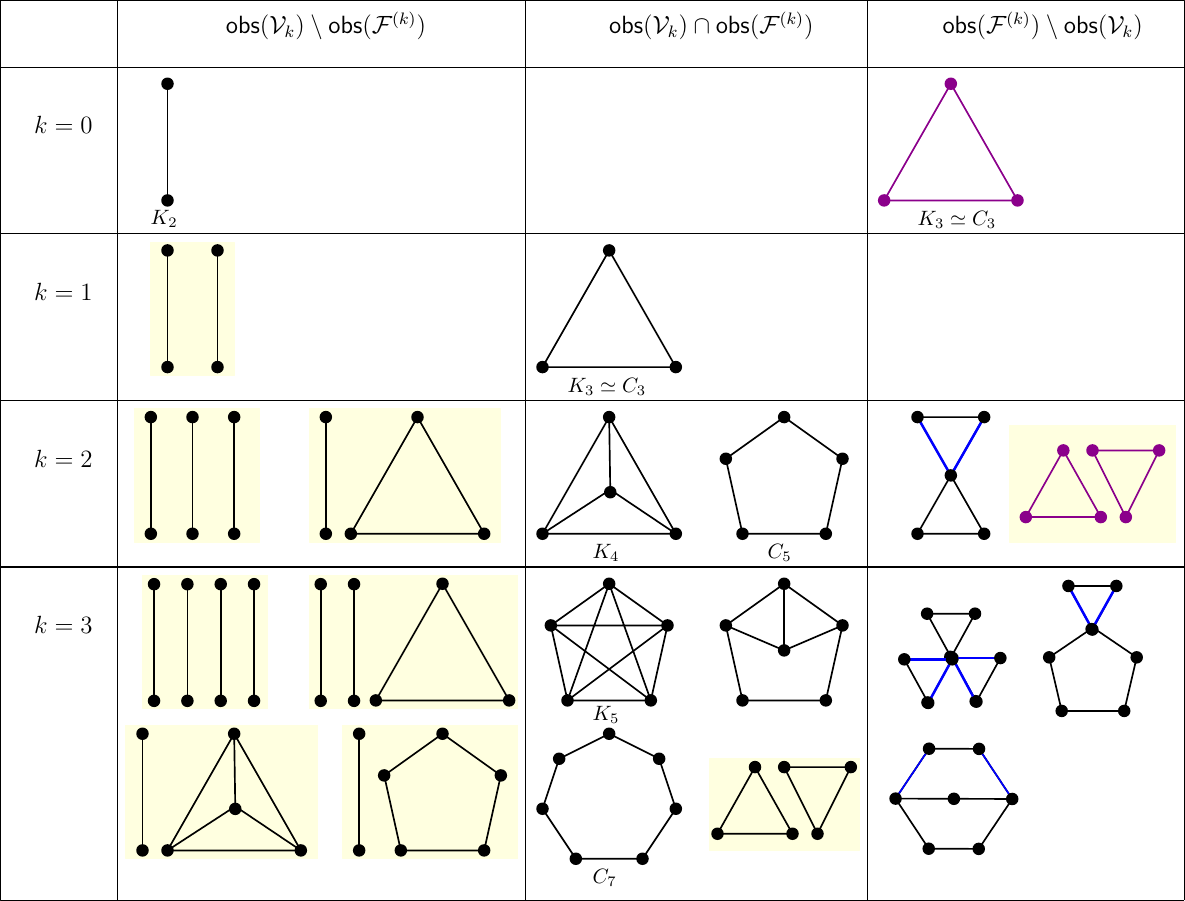}
\caption{The obstructions of $\Vcal_k$ (first and second columns) and $\Fcal^{(k)}$ (second and third columns)  for $k\le3$.
Each graph in $\obs(\Fcal^{(k)})$ is either 1) also a graph in $\obs(\Vcal^{(k)})$ (second column), or 2) can be obtained from a graph in $\obs(\Vcal^{(k)})$ with bridges (first column) by adding edges (in blue in the third column), or 3) is also a graph in $\obs(\Vcal^{(k+1)})$ (in purple in the third column).
We use yellow shadows for disconnected obstructions, to make clear that each of them is a single graph.
}
\label{fig_tab}
\end{figure}

\subsection{Bridges in the obstructions of $\Vcal_k$}\label{sec:VC_obs}

In this subsection, we prove the following.

\begin{lemma}\label{cor_vc_obs}
Let $k\in\bN$ and $G\in\obs(\Vcal_k)$ be a graph.
Then the connected components of $G$ are 2-connected. {Therefore, the bridges of $G$ are isolated edges.}
\end{lemma}

Actually, we prove a more general version of \autoref{cor_vc_obs}  applying on any graph class $\Hcal^{\langle k\rangle}$ defined as follows.
Let $\Hcal$  be a \emph{hereditary} graph class (i.e., closed under vertex deletion) that is  also closed under 1-clique-sums.
Let $\Hcal^{\langle k\rangle}$ be the set of
graphs $G$ such that there exists a set $X\subseteq V(G)$ with $|X|\le k$ and $G-X\in\Hcal$.
In this setting $\Vcal_k=\Ecal^{\langle k\rangle}$,
where $\Ecal$ is  class of edgeless graphs.

We need the following easy lemma.
\begin{lemma}\label{in}
Let $\Hcal$ be a hereditary class, $k\in\bN$ and $H\in\obs(\Hcal^{\langle k\rangle})$. Then, for any $v\in V(H)$, there is a set $S\subseteq V(H)$ of size $k+1$ such that $v\in S$ and $H-S\in\Hcal$.
In particular, $\obs(\Hcal^{\langle k\rangle})\subseteq\Hcal^{\langle k+1\rangle }\setminus \Hcal^{\langle k\rangle}$.
\end{lemma}

\begin{proof}
Let $H\in\obs(\Hcal^{\langle k\rangle})$ and $v\in V(H)$.
By definition of an obstruction, $H\notin\Hcal^{\langle k\rangle}$ and $H-\{v\}\in\Hcal^{\langle k\rangle}$.
So there is a vertex set $S'$ of size at most $k$ in $H-\{v\}$ such that $H-\{v\}-S'\in\Hcal$.
Let $S:=S'\cup\{v\}$.
Then $H-S\in\Hcal$ so $H\in\Hcal^{\langle k+1\rangle}$.
Given that $H\notin\Hcal^{\langle k\rangle}$, we have $|S|>k$, and therefore, $|S|=k+1$.
\end{proof}

Here is the main result of the subsection.

\begin{lemma}\label{lem_VC_obs}
Let $k\in\bN$.
Every connected component of a graph in $\obs(\Hcal^{\langle k\rangle})$ is 2-connected.
\end{lemma}

\begin{proof}
Suppose toward a contradiction that $G\in\obs(\Hcal^{\langle k\rangle})$ has a connected component that is not 2-connected.
Then there is a cut vertex $v$ in $G$.
Let $G_1$ be a connected component of $G-\{v\}$ such that $v\in N_G(V(G_1))$ and let $G_2=G- V(G_1)-\{v\}$.
For $i\in[2]$, let $k_i$ be the minimum $k$ such that $G_i\in \Hcal^{\langle k\rangle}$.
Hence, $G_i\in\Hcal^{\langle k_i\rangle}\setminus\Hcal^{\langle k_i-1\rangle }$.

\begin{claim}\label{cl:1}
$k=k_1+k_2$.
\end{claim}
\begin{proof}
By \autoref{in}, $G\in\Hcal^{\langle k+1\rangle }\setminus \Hcal^{\langle k\rangle}$.
For $i\in[2]$, let $S_i\subseteq V(G_i)$ of size at most $k_i$ be such that $G_i-S_i\in\Hcal$.
Then $S:=S_1\cup S_2\cup\{v\}$ is such that $G-S\in\Hcal$, so $k+1\leq k_1+k_2+1$.

By \autoref{in}, there is a set $S\subseteq V(G)$ of size $k+1$ such that $v\in S$ and $G-S\in\Hcal$.
Given that $\Hcal$ is hereditary, $G_i-(S\cap V(G_i))\in\Hcal$.
Moreover, $G_i\notin\Hcal^{\langle k_i-1\rangle }$ for $i\in[2]$, so we conclude that $|S\cap V(G_i)|\geq k_i$.
Hence, $k+1= |S|=|\{v\}\cup (S\cap V(G_1))\cup (S\cap V(G_2))|\geq k_1+k_2+1.$
\end{proof}

For $i\in[2]$, let $\bar{G_i}:=G[V(G_i)\cup\{v\}]$.
Since $G_i\in\Hcal^{\langle k_i\rangle }\setminus\Hcal^{\langle  k_i-1\rangle }$ and we only add the vertex $v$,
$\bar{G}_i\in\Hcal^{\langle  k_i+1\rangle }\setminus\Hcal^{\langle  k_i-1\rangle}$.

\begin{claim}\label{cl:2}
There is $i\in[2]$ such that $\bar{G}_i\in\Hcal^{\langle k_i+1\rangle }\setminus\Hcal^{\langle k_i\rangle }$.
\end{claim}
\begin{proof}
Suppose that $\bar{G_i}\in\Hcal^{\langle k_i\rangle }$ for $i\in[2]$.
Let $S_i\subseteq V(\bar{G}_i)$ of size $k_i$ be such that $\bar{G}_i-S_i\in\Hcal$.
Then $S:=S_1\cup S_2$ has size at most $k_1+k_2<k+1$.
Moreover, given that $\Hcal$ is closed under 1-clique-sums, we have $G-S\in\Hcal$.
By \autoref{cl:1}, it follows that $G\in\Hcal^{\langle k\rangle}$, a contradiction.
\end{proof}
By \autoref{cl:2}, without loss of generality, we assume that $\bar{G}_1\in\Hcal^{\langle k_1+1\rangle }\setminus\Hcal^{\langle k_1\rangle }$.
Let $G'$ be the graph obtained from the disjoint union of $\bar{G}_1$ and $G_2$.
Given that $\Hcal$ is closed under disjoint union and by \autoref{cl:1}, $G'\in\Hcal^{\langle k_1+1+k_2\rangle }\setminus \Hcal^{\langle k_1+k_2\rangle }=\Hcal^{\langle k+1\rangle}\setminus \Hcal^{\langle k\rangle  }$.
$G'$ is a subgraph of $G$ so
this contradicts the minimality of $G$ as an obstruction of $\Hcal^{\langle k\rangle}$.
\end{proof}

\subsection{Constructing the obstructions of $\Fcal^{(k)}$ from the obstructions of $\Vcal_k$}
\label{sec:id_obs}

What \autoref{lem_bridgeless_obs} and \autoref{lem_VC_obs} tell us is that the
difference (as sets) between $\obs(\Vcal_k)$ and $\obs(\Fcal^{(k)})$ is caused by isolated edges.
Essentially, to go from an obstruction $H$ of $\Vcal_k$ with isolated edges to an obstruction $H'$ of $\Fcal^{(k)}$, we will have to add vertices and edges minimally to get a bridgeless graph.
In this section, we prove that we actually just need to add edges.

Let ${\sf Obs}=\bigcup_{k\in\bN}\obs(\Fcal^{(k)})$.
We have the following easy observation.

\begin{observation}\label{obs-at-most-one}
Let $G\in{\sf Obs}$ and $k:= \idf(G)-1$. Then $G\in\obs(\Fcal^{(k)})$.
\end{observation}

Note that, while we observed in \autoref{in}, in particular, that $\obs(\Vcal_k)\subseteq \Vcal_{k+1}\setminus \Vcal_k$, the same does not hold for $\Fcal_k$.
For instance, $k\cdot K_3$ (see \autoref{fig_univ_obs_idf}) belongs to both $\obs(\Fcal_{2k-2})$ and $\obs(\Fcal_{2k-1})$.
However, we can prove the following.

\begin{lemma}\label{lem+2}
Let $k\in\bN$.
Then $\obs(\Fcal_k)\subseteq \Fcal_{k+2}\setminus \Fcal_k$.
\end{lemma}

\begin{proof}
Let $G\in\obs(\Fcal_k)$ and $uv\in E(G)$.
Let $\Xcal$ be an id-$\Fcal$ partition of $G/uv$.
Then $G\ii\{u,v\}\ii\Xcal\in\Fcal$, so $G\ii\Xcal' \in\Fcal$, where $\Xcal'$ is obtained from $\Xcal$ by further identifying $u$ and $v$.
Thus, $|\bigcup\Xcal'|\le|\bigcup\Xcal|+2\le k+2$, hence the result.
\end{proof}

The main result of this subsection is the following.

\begin{lemma}\label{lem_obs}
Let $G$ be a graph and $k:= \idf(G)-1$.
If $G\in\obs(\Fcal^{(k)})$, then there is
$H\in\obs(\Vcal_k)$ that is a minor of $G$, and for any such $H$, there is $E'\subseteq E(G)$ such that $G-E'=H$.
\end{lemma}

\begin{proof}
By \autoref{obs1}, $G$ is bridgeless.
Therefore, by \autoref{VCeq}, $\idf(G)=\vc(G)$, and thus $G\in\Vcal_{k+1}\setminus\Vcal_k$.
We first prove that, for any edge $e\in E(G)$, $G/e\in\Vcal_k$.

\begin{claim}\label{cl:no_contract}
For any edge $uv\in E(G)$, $G/uv\in\Vcal_k$.
\end{claim}

\begin{cproof}
Suppose toward a contradiction that there is an edge $uv\in E(G)$ such that $G/uv\in\Vcal_{k+1}\setminus\Vcal_k$.
Let $w$ be the heir of $uv$ in $G/uv$.
Since $G\in\obs(\Fcal^{(k)})$, it implies that $G/uv\in\Fcal^{(k)}$.
By \autoref{obs1}, $G$ is bridgeless.
Thus, by \autoref{VCeq} and since $G/uv\in\Fcal^{(k)}\setminus V_k$, it implies that the contraction of $u$ and $v$ created a bridge $e$.
Given that only the edges incident to $u$ and $v$ are involved in the contraction, the
bridges of $G/uv$ are exactly the edges $xw$ where $x\in N_G(u)\cap N_G(v)$ is a cut vertex of $G$ (the edges $xu$ and $xv$ in $G$ are contracted to $xw$ in $G/uv$).
See \autoref{fig_no_contract} for an illustration.
Let $\Ccal$ be the set of all such $x$.
Let $E_1$ be the set of all edges $xu,xv$ of $G$ for $x\in\Ccal$ and let $E_2$ be the set of all edges $xw$ of $G/uv$ for $x\in\Ccal$.

\begin{figure}[h]
\center
\includegraphics[scale=1]{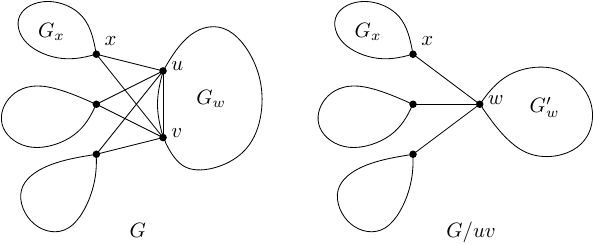}
\caption{Graphs $G$ and $G/uv$.}
\label{fig_no_contract}
\end{figure}
Given that $G/uv\in\Vcal_{k+1}\setminus\Vcal_k$, there is $H\in\obs(\Vcal_k)$ that is a minor of $G/uv$.
For $x\in\Ccal$, let $G_x$ be the connected component of $G-E_1$ containing $x$ and $G_w$ be the disjoint union of the remaining components of $G-E_1$.
Note that $G_x$ is also the connected component of $G/uv-E_2$ containing $x$ for $x\in\Ccal$, and that $G_w':=G_w/uv$ is the union of the other connected components of $G/uv-E_2$.
Given that $G/uv-E_2$ is bridgeless, so are $G_x$ for $x\in\Ccal$ and $G_w'$.
By \autoref{lem_VC_obs}, each connected component of $H$ is 2-connected.
Therefore, given a model $M$ of $H$ of minimal size in $G/uv$, a bridge of $G/uv$ belongs to $M$ if and only if it is an isolated edge in $M$.
Therefore, $H$ is either a minor of $F:=G_w'\cup \bigcup_{x\in\Ccal} G_x$ or, for some $x\in\Ccal$, a minor of $F_x:=G[\{x,w\}]\cup (G_w'-\{w\})\cup(G_x-x)\cup\bigcup_{y\in\Ccal\setminus\{y\}}G_y$.
See \autoref{fig_no_contract2} for an illustration.
\begin{figure}[h]
\center
\includegraphics[scale=1.0]{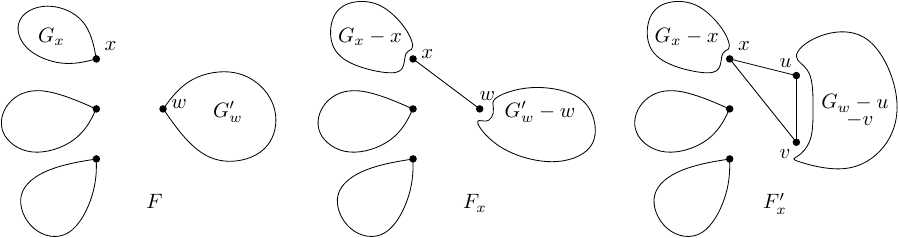}
\caption{Graphs $F$, $F_x$, and $F_x'.$}
\label{fig_no_contract2}
\end{figure}

If $H$ is a minor of $F$ which is a minor of $G$, then $F\in\Vcal_{k+1}\setminus \Vcal_{k}$.
Given that $F$ is bridgeless, by \autoref{VCeq}, we thus have $F\in\Fcal^{(k+1)}\setminus \Fcal^{(k)}$.
This contradicts the fact that $G\in\obs(\Fcal^{(k)})$.

Hence, $H$ is a minor of $F_x$ for some $x\in\Ccal$.
Then $H$ is also a minor of $F_x':=G[\{x,u,v\}]\cup (G_w-u-\{v\})\cup\cup(G_x-x)\bigcup_{y\in\Ccal\setminus\{x\}}G_y$, which is a minor of $G$.
Thus, $F_x,F_x'\in\Vcal_{k+1}\setminus \Vcal_{k}$.
Let $S$ be a vertex cover of $F_x'$ of minimum size, i.e., $|S|=k+1$.
Let $S':=S\cap\{x,u,v\}$.
Given that $G[\{x,u,v\}]$ is a triangle, $|S'|=2$.
But then, $S\setminus S'\cup\{x\}$ is a vertex cover of $F_x$ of size $k$, a contradiction.
\end{cproof}

We now prove that, for any vertex $v\in V(G)$, $G-\{v\}\in\Vcal_k$.

\begin{claim}\label{cl:vtx}
For any vertex $v\in V(G)$, $G-\{v\}\in\Vcal_k$.
\end{claim}
\begin{cproof}
Suppose toward a contradiction that there is a vertex $v\in\Vcal_k$ such that $G-\{v\}\in\Vcal_{k+1}\setminus\Vcal_k$.
If $v$ is an isolated vertex, then $G-\{v\}$ is bridgeless.
So by \autoref{VCeq}, $G-\{v\}\in\Fcal^{(k+1)}\setminus\Fcal^{(k)}$, contradicting that $G\in\obs(\Fcal^{(k)})$.
So there is a vertex $u\in N_G(v)$.
Let us prove that $G/uv\in\Vcal_{k+1}\setminus\Vcal_k$.
This will contradict \autoref{cl:no_contract} and prove the claim.

Suppose toward a contradiction that $G/uv\in\Vcal_k$.
Let $S$ be a vertex cover of $G/uv$ of size $k$.
Let $w$ be the heir of the edge $uv$ in $G/uv$.
If $w$ belongs to $S$, then $S\setminus\{w\}\cup\{u,v\}$ is a vertex cover of $G$ of size $k+1$ containing $v$.
If $w$ does not belong to $S$, then $N_{G/uv}(w)\subseteq S$.
Since $N_{G/uv}(w)=N_G(\{u,v\})$, we conclude that $S\cup\{v\}$ is a vertex cover of $G$ of size $k+1$ containing $v$.
In both cases, $G$ has a vertex cover $S'$ of size $k+1$ containing $v$.
Therefore, $G-\{v\}$ has a vertex cover of size $k$, contradicting the fact that $G-\{v\}\notin\Vcal_k$.
\end{cproof}

Given that $G\in\Vcal_{k+1}\setminus\Vcal_k$, there is $H\in\obs(\Vcal_k)$ that is a minor of $G$.
By \autoref{in}, $H\in\Vcal_{k+1}\setminus\Vcal_k$.
$H$ is obtained from $G$ by a sequence of vertex deletions, edge deletions, and edge contractions such that at each step, the resulting graph belong to $\Vcal_{k+1}\setminus\Vcal_k$.
In particular, we can first do all vertex deletions and edge contractions and then the remaining edge deletions.
But then, by \autoref{cl:no_contract} and \autoref{cl:vtx}, we cannot do any vertex deletion nor edge contraction and still remain in $\Vcal_{k+1}\setminus\Vcal_k$.
Therefore, there is $E'\subseteq E(G)$ such that $G-E'\in\obs(\Vcal_k)$.
This concludes the proof.
\end{proof}

We thus have the following upper bound on the size of obstructions, which is a restatement of \autoref{lem_2}.

\begin{lemma}\label{lem_2-restated}
Let $k\in\bN$.
For any obstruction $G\in\obs(\Fcal^{(k)})$, $|V(G)|\le 2k+4$.
\end{lemma}

\begin{proof}
The obstruction of maximal size in $\obs(\Vcal_k)$ is $(k+1)\cdot K_2$, i.e., the graph obtained from the disjoint union of $k+1$ isolated edges, which has size $2k+2$.

Let $G\in\obs(\Fcal^{(k)})$. By \autoref{lem+2},
we have $\idf(G)\in\{k+1,k+2\}$.
Moreover, by \autoref{lem_obs}, there is $E'\subseteq E(G)$ such that $G-E'\in\obs(\Vcal_{\idf(G)-1})$.
Therefore, $G-E'$, and thus $G$, has size at most $2k+4$.
\end{proof}

\section{Universal obstruction}
\label{sec:univ_obs}

As we already explained in the introduction, we denote by $k\cdot K_{3}$ the union of $k$ disjoint copies of $K_{3}$ and by $C_{k}$ the cycle on $k$ vertices.
Recall also that the \emph{marguerite} of order $k$, denoted by $k* K_3$,
is the graph $(k\cdot K_3)\ii X$, where $X\subseteq V(k\cdot K_3)$ is a set containing exactly one vertex from each triangle (see \autoref{fig_univ_obs_idf}).

\begin{figure}[h]
\center
\includegraphics{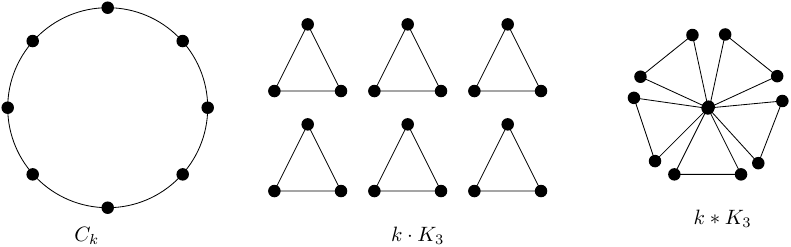}
\caption{The universal obstruction for {\sc Identification to Forest}.}
\label{fig_univ_obs_idf}
\end{figure}

\begin{observation}\label{obs_univ_obs}
$C_{2k+1}$, $\lfloor \frac{k}{2}+1\rfloor\cdot K_3$, and $(k+1)*K_3$ are in $\obs(\Fcal^{(k)})$.
\end{observation}

\begin{lemma}\label{th_univ_obs_idf}
If $G$ excludes every graph in $\{C_k,k\cdot K_3,k*K_3\}$ as a minor, then $\idf(G)=\Ocal(k^4\cdot\log^2 k)$.
\end{lemma}

\begin{proof}
Let $G$ be a $\{C_k,k\cdot K_3,k*K_3\}$-minor-free graph.
By \autoref{VCeq}, we can assume without loss of generality that $G$ is bridgeless.
In particular, any vertex of $G$ has degree at least two.

By the Erdős–Pósa's theorem \cite{ErdosP65inde}, either $G$ has a packing of $k$ cycles, or there is a set $X$ of size $\Ocal(k\cdot \log k)$ such that $G-X\in\Fcal$.
Given that $G$ is $k\cdot K_3$-minor-free, there exists such a set~$X$ and
$G[X]$ has at most $\Ocal(k\cdot \log k)$ connected components.

Let $C$ be a connected component of $G[X]$. 
Let $\Tcal_C$ be the set of trees in $F$ with a neighbor in $C$. 
Given that $G$ is bridgeless and that any path from a vertex of $T\in\Tcal_C$ to a vertex of $G-V(C)- V(T)$ intersects $C$, we have $|E_G(V(T),V(C))|\geq 2$.
Hence, there is a cycle in the graph induced by $T$ and $C$.
Hence, $|\Tcal_C|*K_3$ is a minor of $G$.
Therefore, $|\Tcal_C|\leq k-1$.

Let $T\in\Tcal_C$.
Let $T^C$ be the subtree of $T$ obtained by iteratively removing every leaf of $T$ that is not in $N_G(C)$.
Hence, for every pair of leaves $u,v$ of $T^C$, there are two $(u,v)$-paths $P_1$ and $P_2$, the first one in $T^C$ and the second one going through $C$, that are internally vertex-disjoint.
So there is a cycle of length at least $\Delta(T^C)+1$, where $\Delta(T^C)$ denotes the diameter of $T^C$.
Given that $C_{k}$ is not a minor of $G$, $T^C$ has diameter at most $k-2$.

Let $L(T^C)$ denote the leaves of $T^C$, and let $PL(T^C)$ denote the parents of vertices in $L(T^C)$.
We claim that $|PL(T^C)|\leq k$.
Indeed, let $u\in L(T^C)$ be a leaf picked arbitrarily.
Let $V'=V(C)\cup V(T^C)\setminus L(T^C)\setminus PL(T^C)\cup\{u,{\sf p}(u)\}$, where ${\sf p}(u)$ is the parent of $u$ in $T^C$.
Observe that, since $u$ is connected to $C$, $G[V']$ is connected.
Hence, we can contract $V'$ to a single vertex $c$ to obtain a graph $G'$.
For each $t\in PL(T^C)\setminus\{{\sf p}(u)\}$, there is a triangle $ctv_t$ where $v_t\in L(T^C)$ is a child of $t$.
Hence, $(|PL(T^C)|-1)*K_3$ is a subgraph of $G'$ and thus a minor of $G$.
Since $k*K_3$ is not a minor of $G$, we proved our claim.

Therefore, $|V(T^C)\setminus L(T^C)|\leq\Delta(T^C)\cdot|PL(T^C)|\leq k\cdot(k-2)$.

Let $E'$ be the set of all edges of $F$ that do not belong to $T^C$ for any $C\in\cc(G[X])$ and $T\in\Tcal_C$.
Let $e\in E'$.
Since $e$ is not a bridge, $e$ is part of a cycle $C_e$.
Hence, there are $C,C'\in\cc(G[X])$ and $T\in\Tcal_C\cap\Tcal_{C'}$ such that any path from $T^C$ to $T_{C'}$ in $T$ goes through $e$.
Moreover, there are at most $k-5$ such edges between $T^C$ to $T_{C'}$, since otherwise $C_e$ would have length at least $k$.
Hence, $|E'|\leq (k-5)\cdot\binom{|\cc(G[X])|}{2}\cdot\max_{C\in\cc(G[X])}|\Tcal_C|=\Ocal(k^4\cdot\log^2 k)$.

Let $V'\subseteq V(G)$ be the union of $X$, of the endpoints of edges in $E'$, and of the internal nodes of $T^C$ for any $C\in\cc(G[x])$ and any $T\in\Tcal_C$.
Then, $V(G)\setminus V'\subseteq L(F)$, so $G\ii V'$ is a star.
Moreover, $|V|= \Ocal(k\cdot\log k+ k^4\cdot\log^2 k+ k\cdot\log k\cdot k\cdot k^2)=\Ocal(k^4\cdot\log^2 k).$
\end{proof}

\begin{proof}[Proof of \autoref{th_univ_obs_idf_more}]
The first condition of the universal obstruction property
follows from \autoref{th_univ_obs_idf} an the second one follows from  \autoref{obs_univ_obs}.
\end{proof}

\section{Discussion and open problems}
\label{sec-discussion-and-open}

In this paper we initiated the study of graph modification problems where the modification operation is vertex identification.
We defined the problem {\sc Identification to $\Hcal$} and studied the
case where the target class $\Hcal$ is the class of forest, denoted by $\Fcal$.

\paragraph{Relation with {\sc Contraction to $\Hcal$}.}
An important feature of  {\sc Identification to $\Hcal$} is that
it behaves similarly to the problem   {\sc Deletion to $\Hcal$},
in the sense that both problems are {\sf FPT} when $\Hcal$ is a minor-closed graph class. This follows from \autoref{lem_minor} and the algorithmic consequence of the Robertson and Seymour's theorem \cite{RobertsonS12XXII,RobertsonS09XXI,RobertsonS95GMXIII,KawarabayashiKR11thed,korhonen2024minor}.
 It is easy to observe that the problem
{\sc Contraction to $\Hcal$} (that is, asking whether $k$ edge contractions yield property $\Hcal$) does not
have this property. To see this, let $\Pcal$ be the class of planar graphs and let $K_{3,4}^+$
(resp. $K_{2,3}^+$)
be the graph obtained from $K_{3,4}$ (resp. $K_{2,3}$) by adding an edge $e$ between two vertices   of degree three (resp. two).
Contracting $e$ yields a planar (resp. acyclic) graph, so $(K_{3,4}^+,1)$ (resp. $(K_{2,3}^{+},1)$) is a \yes-instance of {\sc Contraction to $\cal P$} ({\sc Contraction to Forest}). However, $(K_{3,4},1)$ (resp. $(K_{2,3},1)$) is a \no-instance of the corresponding  problem.

Let us define the parameter ${\sf ec}_{\Hcal}:\gall\to\Nbbb$,
corresponding to the problem  {\sc Contraction to $\Hcal$}, i.e.,  ${\sf ec}_{\Hcal}(G)$ is the minimum number of edge contractions that can transform $G$ to graph in $\Hcal$. As we observed above, neither ${\sf ec}_{\Fcal}$ nor  ${\sf ec}_{\Pcal}$ are minor-monotone, and similar counterexamples can be found for other instantiations of $\Hcal$.
We use ${\ecf}$ as a shortcut for ${\sf ec}_{\Fcal}$ and we next observe that $\idf$ and
${\ecf}$ are functionally equivalent.

\begin{lemma}
\label{oupe4loi}
For every graph $G$ it holds that $\idf(G)=\Ocal(\ecf(G))$ and that $\ecf(G)=\Ocal((\idf(G))^3)$.
\end{lemma}

\begin{proof}
Using the fact that edge contractions
are also edge identifications, it easily follows that, for every graph $G$,
$\idf(G)≤2\cdot \ecf(G)$.

Assume now that  $\idf(G)≤k$ and we claim that
$\ecf(G)=\Ocal(k^3)$.
To prove this claim  we first observe  that, because $\idf(k\cdot K_{3})=\Omega(k)$ and $\idf(k * K_{3})=\Omega(k)$ (see \autoref{obs_univ_obs}),  it follows that
the number of 2-connected components of $G$  that are not bridges
is bounded by some linear function of $k$.
Let $B$ be a 2-connected component of $G$.
As $B$ is a minor of $G$, it has an id-$\Fcal$ partition $\Xcal=\{X_1,\dots,X_p\}$ of order $≤k$.
For $i\in[p]$, let  $x_{1}^{i},\ldots,x_{p_{i}}^{i}$ be an ordering of the
vertices of $X_i$ and let $F_{i}=\{\{x_{1}^{i},x_{2}^{i}\},\{x_{2}^{i},x_{3}^{i}\},\ldots,\{x_{p_{i}-1}^{i},x_{p_{i}}^{i}\}\}$. 
Let also $F=F_{1}\cup\cdots\cup F_{p}$.
Clearly, the 2-element sets in $F$ are not necessarily
edges of $B$. For each $\{x,y\}\in F$ we define a set of edges $F_{x,y}$
as follows. As $B$ is 2-connected, $x$ and $y$
belong to a cycle of $B$. As $\idf(C_{k})=\Omega(k)$ (see \autoref{obs_univ_obs}), this implies that
$x$ and $y$ are joined in $B$ by a path of length $\Ocal(k)$.
The edges of this path are the edges in $F_{x,y}$.
We now set $F^+=\bigcup_{\{x,y\}\in F} F_{x,y}$
and observe that  $|F^+|=\Ocal(k^2)$. Notice now that
contracting the edges of $F^+$ in $B$ yields an acyclic graph.
Therefore, applying these contractions to every non-bridge connected component of $G$, we obtain an acyclic graph. As there are $\Ocal(k)$ such components, the lemma follows.
\end{proof}

In other words, $\ecf$ is not minor-monotone but, however,
it is ``functionally'' monotone in the sense that if $G'$ is a minor
of $G$ then $\ecf(G')≤\Ocal((\ecf(G))^3)$.\footnote{The cubic bound in \autoref{oupe4loi} is just indicative and has not been optimized.}
While it is easy to see that $\id_{\Hcal}(G)≤2\cdot \ec_{\Hcal}(G)$, we also conjecture that  an upper bound as the one of \autoref{oupe4loi} holds for every minor-closed class $\Hcal$.

\begin{conjecture}
For every minor-closed graph class $\Hcal$, there is a function
$f_{\Hcal}:\Nbbb\to\Nbbb$ such that for every $G$, ${\sf ec}_{\Hcal}≤f_{\Hcal}(\id_{\Hcal}(G))$.
\end{conjecture}

Note that {\sc Contraction to $\Hcal$} is known to be ${\sf W}[1]$-hard, parameterized by the solution size, for several families $\Hcal$ that are {\sl not} minor-closed, such as chordal graph or split graphs (see~\cite{AgrawalLSZ19} and the references cited therein). However, when $\Hcal$ is minor-closed, the recent meta-algorithmic results in~\cite{FominGSST23comp} (further generalized in~\cite{SauST2024parame}) imply that {\sc Contraction to $\Hcal$} is (constructively) \FPT (see~\cite{biclique-contraction-23,HeggernesHLLP14contra} for explicit algorithms for some particular families).
Also, as it has been proved in \cite{HeggernesHLLP14contra}, {\sc Contraction to Forest} is not expected to admit a polynomial kernel.  Interestingly, the kernelization we give in this paper for {\sc Identification to Forest}, under the light of the polynomial-gap functional equivalence  of \autoref{oupe4loi}, can be seen as some kind of ``functional kernel'' for {\sc Contraction to Forest}.

\paragraph{Identification minors.} We say that a graph $H$ is
an \emph{identification minor} of  a graph $G$
if $H$ can be obtained from a minor of $G$ after identifying vertices.
As the minor relation between two graphs also implies their identification minor relation, Robertson and Seymour's theorem~\cite{RobertsonS04XX} implies that graphs are well-quasi-ordered by the identification minor relation.
It is also easy to observe that, for every graph $H$,
the graphs in the set ${\Mcal}_{H}$ of minor-minimal graphs containing
$H$ as an identification have size is bounded by a quadratic function of $|H|$.
Therefore, checking whether
$H$ is an identification minor of $G$ can be done in time $\Ocal_{|H|}(|G|^{1+\varepsilon})$, according to the recent results in \cite{korhonen2024minor}.

It is a natural question to ask whether graphs are well-quasi-ordered with respect to the vertex identification operation alone. The answer turns out to be negative. Indeed, there is  an infinite antichain $(H_k)_{k\in\bN}$, where
 $H_k$ is the graph formed from a cycle on $3k$ vertices $p_1,\ldots,p_{3k}$ by adding three vertices $a_1,a_2,a_3$ and an edge between each pair $(a_i,p_j)$ such that $j$ is equal to $i$ modulo three.
See \autoref{fig_antichain_id} for an illustration. It can be verified that this family of graphs is indeed an antichain, even if we allow both vertex identifications and vertex removals.

\begin{figure}[h]
\center
\includegraphics[scale=0.8]{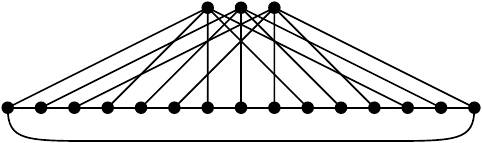}
\caption{The graph $H_k$ for $k=5$. We give credit to Hugo Jacob for finding it.}
\label{fig_antichain_id}
\end{figure}

We now wish to give the following interpretation of \autoref{th_univ_obs_idf} in terms of identification minors.
To prove it, one needs to observe that  $k*K_3$ is an identification-minor of both $k\cdot K_3$ and $C_{3k}$.

\begin{theorem}
For every graph $G$ and positive integer $k$, either  $G$ contains the $k$-marguerite $k* K_{3}$
as an identification  minor, or $G$ can become acyclic after applying $\Ocal(k^4\cdot\log^2k)$ vertex identifications.
\end{theorem}

The above theorem can be seen as an analogue
of the Erdős-Pósa's theorem~\cite{ErdosP65inde} where instead of the vertex removal operation we have vertex identification, and instead of $k\cdot K_{3}$ minor containment we have $k* K_{3}$ identification minor containment. Which are the (parametric) graphs appearing as identification minors when $\id_{\Hcal}$ is big enough,  for a general minor-closed graph class $\Hcal$?

 \end{document}